\documentclass[preprint,sort&compress]{elsarticle}

\usepackage{array}
\usepackage{longtable}
\usepackage{amssymb,amsmath}

\journal{J. of Geometry and Physics}

\newtheorem{propos}{Proposition}
\newproof{proof}{Proof}
\newdefinition{remark}{Remark}
\newtheorem{theorem}{Theorem}
\newtheorem{lemma}{Lemma}

\newcommand {\ds}{\displaystyle}

\def\sgrad{\operatorname{sgrad}}
\def\rank{\operatorname{rank}}
\def\trace{\operatorname{trace}}

\usepackage{hyperref}

\begin{document}

\begin{frontmatter}
\author[fin]{P. E. Ryabov\corref{cor1}}
\ead{orelryabov@mail.ru}

\address[fin]{Financial University, Leningradsky Avenue, 49, Moscow, 125993, Russia}

\cortext[cor1]{Corresponding author}

\title{New invariant relations for the generalized two-field gyrostat\tnoteref{t1}}
\tnotetext[t1]{This work is partially supported by the grants of RFBR and Volgograd Region Authorities No. 14-01-00119 and 13-01-97025.}

\begin{abstract}
In the paper, we consider a completely integrable Hamiltonian system with three degrees of freedom found by V.V.\,Sokolov and A.V.\,Tsiganov. This system is known as the generalized two-field gyrostat. For the case of only gyroscopic forces present, we find new invariant four-dimensional submanifolds such that the induced dynamical systems are almost everywhere Hamiltonian with two degrees of freedom.
\end{abstract}


\end{frontmatter}

\tableofcontents

\section{Introduction}

Motion of the generalized two-field gyrostat is described by the following
system of differential equations:
\begin{equation}\label{eq_1}
\begin{array}{l}
 \ds{\dot{\boldsymbol M}={\boldsymbol M}\times\frac{\partial H}{\partial{\boldsymbol
M}}+ {\boldsymbol\alpha}\times\frac{\partial H}{\partial{\boldsymbol\alpha}}+
{\boldsymbol\beta}\times\frac{\partial H}{\partial{\boldsymbol\beta}},}\\[5mm]
\ds{\dot{\boldsymbol \alpha}={\boldsymbol\alpha}\times\frac{\partial
H}{\partial{\boldsymbol M}},\quad \dot{\boldsymbol
\beta}={\boldsymbol\beta}\times\frac{\partial H}{\partial{\boldsymbol M}},}
\end{array}
\end{equation}
with the Hamiltonian function \cite{SokTsi02}
\begin{equation}\label{eq_2}
\begin{array}{l}
H_{\varepsilon_1,\varepsilon_2}= M_1^2+M_2^2+2M_3^2+2\lambda M_3-2\varepsilon_2(\alpha_1+\beta_2)\\[3mm]
\qquad +2\varepsilon_1(M_2\alpha_3-M_3\alpha_2+M_3\beta_1-M_1\beta_3).
\end{array}
\end{equation}
Here ${\boldsymbol M}, {\boldsymbol\alpha}, {\boldsymbol\beta}$ stand for the total kinetic momentum and the intensities of the two forces considered in the moving frame formed by the principal inertia axes of the body. The gyrostatic momentum is directed along the dynamic symmetry axis and its axial component is denoted by $\lambda$. The parameters $\varepsilon_1,\varepsilon_2$ are called the deformation parameters since their zero values define important partial cases and the connections of the problem with some previously known integrable cases.

Treating $\mathbb{R}^9=\{({\boldsymbol M}, {\boldsymbol\alpha}, {\boldsymbol\beta})\}$ as the Lie coalgebra $e(3,2)^*$ we obtain the Lie--Poisson bracket
\begin{equation}\label{eq_3}
\begin{array}{l}\{M_i,M_j\}=\varepsilon_{ijk}M_k, \quad \{M_i,\alpha_j\}=\varepsilon_{ijk}\alpha_k,\quad
\{M_i,\beta_j\}=\varepsilon_{ijk}\beta_k,\\[5mm]
\{\alpha_i,\alpha_j\}=0, \quad \{\alpha_i,\beta_j\}=0, \quad \{\beta_i,\beta_j\}=0, \\[5mm]
\varepsilon_{ijk}=\frac{1}{2}(i-j)(j-k)(k-i),\quad 1\leqslant i,j,k\leqslant 3.
\end{array}
\end{equation}
With respect to this bracket, system (\ref{eq_1}) can be represented in the Hamiltonian form
\begin{equation*}
\dot x=\{H_{\varepsilon_1,\varepsilon_2},x\}
\end{equation*}
for any coordinate function $x$ on $\mathbb{R}^9$.

Note that the Casimir functions of the bracket (\ref{eq_3}) are ${\boldsymbol\alpha}^2$,
${\boldsymbol\alpha}\cdot{\boldsymbol\beta}$, and ${\boldsymbol\beta}^2$.
Therefore we define the phase space $\cal P$ of system \eqref{eq_1} as a common
level of these functions
\begin{equation}\label{eq_4}
\boldsymbol\alpha^2=a^2,\quad \boldsymbol\beta^2=b^2,\quad
{\boldsymbol\alpha}\cdot{\boldsymbol\beta}=c, \quad (0<b<a, |c|<ab).
\end{equation}

In \cite{SokTsi02}, for system \eqref{eq_1} with the Hamiltonian function \eqref{eq_2},
V.\,V.\,Sokolov and A.\,V.\,Tsiganov gave a Lax representation with a spectral parameter
and thereby proved Liouville complete integrability of this system. This Lax representation
generalizes the $L$-$A$ pair for the Kowalevski gyrostat in a double field found by A.\,G.\,Reyman and
M.\,A.\,Semenov-Tian-Shansky \cite{ReySem1987}.

In this paper, we write out the general integrals in involution for the Hamiltonian function \eqref{eq_2}, obtain the explicit equation of the algebraic curve associated with the Lax pair of Sokolov and Tsiganov and show how to use this curve to construct some special surfaces generating  so-called critical subsystems in the considered integrable system.
After that, we mostly restrict ourselves to the Hamiltonian function without linear potential, i.e., to the case $\varepsilon_2=0$. For the critical subsystems arising in this case we obtain the description of the invariant submanifolds in terms of invariant relations and reveal some important features of the critical subsystems, such as degeneration of the induced symplectic structure and the types of critical points with respect to the initial system.

\section{The notion of a critical subsystem}

The problem of motion of the generalized two-field gyrostat restricted the phase space (\ref{eq_4}) is an integrable Hamiltonian system with three degrees of freedom.

In the study of global analytical and topological features of systems with three degrees of freedom, critical subsystems are of special interest. They form the critical sets of integral mappings, define the stratification of the phase space and the bifurcations of the integral manifolds.
The notion of a critical subsystem was formed in the works by M.\,P.\,Kharlamov \cite{Kh2005,Kh2009}. In this section, we follow the general approach described in \cite{KhRCDNew}. The idea of a critical subsystem is as follows.

Let $K$ and $G$ be two integrals in involution of a Liouville integrable Hamiltonian system with three degrees of freedom on $\cal P$ ($\dim {\cal P}=6$) with a Hamiltonian function $H$. In what follows, having a general first integral denoted by an uppercase letter we denote by the corresponding lowercase letter its particular value (the integral constant).
We define
\begin{equation*}
{\cal F}: {\cal P}\to {\mathbb R}^3
\end{equation*}
by ${\cal F}(x) =\{g = G(x), k = K(x), h = H(x)\}$. The mapping $\cal F$ is called the \textit{momentum mapping}.
By $\cal C$ we denote the set of all critical points of $\cal F$, i.e., the set of points $x$ such that $\rank d{\cal F}(x) < 3$. The set of critical values $\Sigma = {\cal F}({\cal C}) \subset {\mathbb R}^3$ is called the \textit{bifurcation diagram}. Normally, $\Sigma$ is a stratified 2-manifold.

Let
\begin{equation}\label{eq_5}
{\cal L}(h,k,g)=0
\end{equation}
be the equation of a two-dimensional surface $\Pi_{\cal L}$ that contains one of the smooth two-dimensional leaves of $\Sigma$. We call such surfaces the bifurcation surfaces. Thus, the closure of the 2-skeleton of the bifurcation diagram is a subset of the union of the bifurcation surfaces.
Introduce the function
\begin{equation}\label{eq_6}
\Phi_{\cal L}={\cal L}\circ{\cal F}: {\cal P}\to {\mathbb R}.
\end{equation}
Then the corresponding \textit{critical subsystem} ${\cal M}_{\cal L}$ is defined as the closure of the set of critical points of rank $2$ on the zero level of the integral $\Phi_{\cal L}$. Obviously, ${\cal M}_{\cal L}$ is an invariant subset in $\cal P$, consisting of critical points of the mapping ${\cal F}$. The critical subsystem ${\cal M}_{\cal L}$ can be described by the following system of equations:
\begin{equation*}
\Phi_{\cal L}=0,\quad d\Phi_{\cal L}=0.
\end{equation*}
Almost everywhere on ${\cal M}_{\cal L}$ this system has rank 2, thus, locally ${\cal M}_{\cal L}$ is defined by two equations.

Recall a well-known fact from symplectic geometry \cite{Fom1995}.
\begin{lemma}\label{lem1_1}
Suppose a submanifold ${\cal M}$ of a symplectic manifold ${\cal P}$ is defined by the
system of independent equations
\begin{equation}\label{eq_7}
f_1 = 0,\quad f_2 = 0.
\end{equation}
Then the $2$-form on ${\cal M}$ induced by the symplectic structure of ${\cal P}$ degenerates exactly
on the set
\begin{equation*}
\{f_1, f_2\} = 0.
\end{equation*}
\end{lemma}
Since critical subsystems are usually described by the systems of the form \eqref{eq_7},
the induced symplectic structure may degenerate on a set of codimension $1$. In this
case the subsystem is said to be \textit{almost Hamiltonian}.

Denote by $X_H$ the considered Hamiltonian vector field on $\cal P$. We suppose that the integrable system $X_H$ is non-degenerate in the Arnold sense. The following lemma \cite{Kh2007} gives a useful tool to verify whether the common level of two functions \eqref{eq_7} consists of critical points of the momentum mapping $\cal F$.

\begin{lemma}\label{lem1_2}
Consider a system of equations \eqref{eq_7} in some domain $W$ open in ${\cal P}$. Let $\cal F$ be the momentum mapping of the Hamiltonian vector field $X_H$ on $\cal P$. Suppose that $\cal M \subset W$ is defined by system \eqref{eq_7}. We also suppose that

(i) $f_1$, $f_2$ are smooth functions that are independent on $\cal M$;

(ii) $X_H f_1 = 0$, $X_H f_2 = 0$ on $\cal M$;

(iii) almost everywhere on $\cal M$ the Poisson bracket $\{f_1, f_2\}\ne 0$.

\noindent Then $\cal M$ consists of critical points of the mapping $\cal F$.
\end{lemma}

In what follows, we find new critical subsystems of the generalised two-field gyrostat (invariant almost everywhere four-dimensional submanifolds with the induced dynamical systems being almost everywhere Hamiltonian systems with two degrees of freedom). As an application, we obtain the type (elliptic or hyperbolic) of the points of these critical subsystems considered as critical points of the initial system on the whole phase space ${\cal P}$.

\section{How to find the equations of the bifurcation surfaces?}

In this section we show that the equations of surfaces of the type \eqref{eq_5} (implicit or parametric) could be obtained as the equations for the discriminant sets of some polynomials corresponding to singularities of the algebraic curve ${\cal E}(z,\zeta)$ associated with the Lax representation.

For the Hamiltonian function \eqref{eq_2}, we represent the additional integrals
$K_{\varepsilon_1,\varepsilon_2}$ and $G_{\varepsilon_1,\varepsilon_2}$ as integrals that depend on two deformation parameters
$\varepsilon_1,\varepsilon_2$:
\begin{equation*}
\begin{array}{l}
K_{\varepsilon_1,\varepsilon_2}=Z_1^2+Z_2^2-\lambda[(M_3+\lambda)(M_1^2+M_2^2)+2\varepsilon_2(\alpha_3M_1+\beta_3M_2)]\\[3mm]
\qquad+\lambda\varepsilon_1^2({\boldsymbol\alpha}^2+{\boldsymbol\beta}^2)M_3+2\lambda\varepsilon_1[\alpha_2M_1^2-\beta_1M_2^2-(\alpha_1-\beta_2)M_1M_2]
-2\lambda\varepsilon_1^2\omega_\gamma,\\[3mm]
G_{\varepsilon_1,\varepsilon_2}=\omega_\alpha^2+\omega_\beta^2+2(M_3+\lambda)\omega_\gamma-
2\varepsilon_2({\boldsymbol\alpha}^2\beta_2+{\boldsymbol\beta}^2\alpha_1)\\[3mm]
\qquad +2\varepsilon_1[{\boldsymbol\beta}^2(M_2\alpha_3-M_3\alpha_2)-
{\boldsymbol\alpha}^2(M_1\beta_3-M_3\beta_1)]\\[3mm]
\qquad+2({\boldsymbol\alpha}\cdot{\boldsymbol\beta})[\varepsilon_2(\alpha_2+\beta_1)+\varepsilon_1(\alpha_3M_1-\alpha_1M_3+\beta_2M_3-\beta_3M_2)].
\end{array}
\end{equation*}

Here we use the following notation:
\begin{equation*}
\begin{array}{l}
Z_1=\frac{1}{2}(M_1^2-M_2^2)+\varepsilon_2(\alpha_1-\beta_2)\\[3mm]
\qquad+\varepsilon_1[M_3(\alpha_2+\beta_1)-M_2\alpha_3-M_1\beta_3]+
\frac{1}{2}\varepsilon_1^2({\boldsymbol\beta}^2-{\boldsymbol\alpha}^2),\\[3mm]
Z_2=M_1M_2+\varepsilon_2(\alpha_2+\beta_1)-\varepsilon_1[M_3(\alpha_1-\beta_2)+\beta_3M_2-\alpha_3M_1]-\varepsilon_1^2(
{\boldsymbol\alpha}\cdot{\boldsymbol\beta}),\\[3mm]
\omega_\alpha=\alpha_1M_1+\alpha_2M_2+\alpha_3M_3,\quad
\omega_\beta=\beta_1M_1+\beta_2M_2+\beta_3M_3,\\[3mm]
\omega_\gamma=M_1(\alpha_2\beta_3-\beta_2\alpha_3)+M_2(\alpha_3\beta_1-\alpha_1\beta_3)+M_3(\alpha_1\beta_2-\alpha_2\beta_1).
\end{array}
\end{equation*}

In the special case $\varepsilon_1=0,\varepsilon_2=1$, we get the integrals for the
problem of the Kowalevski gyrostat motion in two homogeneous fields
\cite{ReySem1987,BobReySem1989}.

Let for brevity $H=H_{\varepsilon_1,\varepsilon_2}$, $K=K_{\varepsilon_1,\varepsilon_2}$, and $G=G_{\varepsilon_1,\varepsilon_2}$. For the Lax pair of Sokolov and Tsiganov \cite{SokTsi02}, the spectral curve ${\cal E}(z,\zeta)$ has the equation
\begin{equation}\label{eq_8}
{\cal E}(z,\zeta)\,:\, \, d_4\zeta^4+d_2\zeta^2+d_0=0,
\end{equation}
where
\begin{equation*}
\begin{array}{l}
d_4=-z^4-\varepsilon_1^2({\boldsymbol\alpha}^2+{\boldsymbol\beta}^2)z^2-
\varepsilon_1^4[{\boldsymbol\alpha}^2{\boldsymbol\beta}^2-({\boldsymbol\alpha}\cdot{\boldsymbol\beta})^2],\\[3mm]
d_2=2z^6+[\varepsilon_1^2({\boldsymbol\alpha}^2+{\boldsymbol\beta}^2)-h-\lambda^2]z^4+
[\varepsilon_2^2({\boldsymbol\alpha}^2+{\boldsymbol\beta}^2)-
\varepsilon_1^2g]z^2\\[3mm]
\qquad+2\varepsilon_1^2\varepsilon_2^2[{\boldsymbol\alpha}^2{\boldsymbol\beta}^2-({\boldsymbol\alpha}\cdot{\boldsymbol\beta})^2],\\[3mm]
d_0=-z^8+hz^6+f_{\varepsilon_1,\varepsilon_2}z^4+\varepsilon_2^2gz^2-
\varepsilon_2^4[{\boldsymbol\alpha}^2{\boldsymbol\beta}^2-({\boldsymbol\alpha}\cdot{\boldsymbol\beta})^2].
\end{array}
\end{equation*}
The most complicated coefficient $f_{\varepsilon_1,\varepsilon_2}$ at $z^4$ in $d_0$ is expressed in terms of the integral constants $h$, $k$, and $g$ as follows:
\begin{equation*}
f_{\varepsilon_1,\varepsilon_2}=\varepsilon_1^2g+k-\varepsilon_1^4({\boldsymbol\alpha}\cdot{\boldsymbol\beta})^2-
\frac{1}{4}[h^2+2\varepsilon_1^2({\boldsymbol\alpha}^2+{\boldsymbol\beta}^2)h+\varepsilon_1^4({\boldsymbol\alpha}^2-{\boldsymbol\beta}^2)^2]-
\varepsilon_2^2({\boldsymbol\alpha}^2+{\boldsymbol\beta}^2).
\end{equation*}
The curve \eqref{eq_8} can be considered as zero level of the mapping ${\cal
E}:\overline{{\mathbb C}}\times\overline{{\mathbb C}}\to\overline{{\mathbb C}}$. Denote by
$\widetilde{\Sigma}\subset{\mathbb R}^3(g,k,h)$ the set of such values of the integral constants for which zero is a critical value
of the mapping $\cal E$.

The experience of studying Hamiltonian systems \cite{Kh2005}, \cite{Kh_Ryab1997} shows that
\begin{equation*}
\Sigma\subset\widetilde{\Sigma}
\end{equation*}
and the bifurcation set $\Sigma$ is cut out of $\widetilde{\Sigma}$ by the requirement that the initial variables are real. In turn, the set $\widetilde{\Sigma}$ at finite points of $\overline{{\mathbb C}}\times\overline{{\mathbb C}}$ is defined by the system
\begin{equation}\label{eq_9}
{\cal E}(z,\zeta)=0,\quad \frac{\partial}{\partial z}{\cal E}(z,\zeta)=0,\quad
\frac{\partial}{\partial \zeta}{\cal E}(z,\zeta)=0.
\end{equation}
For the algebraic curve \eqref{eq_8} system \eqref{eq_9} leads to two different possibilities, either
\begin{equation}\label{eq_10}
d_0=0 \quad \mathrm{and} \quad \frac{\partial}{\partial z} d_0 =0,
\end{equation}
or
\begin{equation}\label{eq_11}
{\cal D}=0 \quad \mathrm{and} \quad \frac{\partial}{\partial z} {\cal D}=0,
\end{equation}
where ${\cal D}=d_2^2-4d_4d_0$.
Obviously, these systems define in ${\mathbb R}^3(g,k,h)$ the surfaces of multiple roots (the discriminant surfaces) of the polynomials $d_0=d_0(z)$ and ${\cal D}={\cal D}(z)$.

From \eqref{eq_10} putting $t=z^2$, we can represent the first discriminant surface in the parametric form
\begin{equation}\label{eq_12}
\begin{array}{l}
\ds{g(t)=\frac{ht^2-2t^3}{\varepsilon_2^2}+\frac{2\varepsilon_2^2(a^2b^2-c^2)}{t},}\\[3mm]
\ds{k(t)=3t^2-2ht+\frac{\varepsilon_1^2(2t^3-ht^2)}{\varepsilon_2^2}+\frac{\varepsilon_2^2(c^2-a^2b^2)(2\varepsilon_1^2t+\varepsilon_2^2)}{t^2}}\\[3mm]
\qquad \ds{+\frac{1}{4}\{h^2+2\varepsilon_1^2(a^2+b^2)h+\varepsilon_1^4[(a^2-b^2)^2+4c^2]+4\varepsilon_2^2(a^2+b^2)\}}.
\end{array}
\end{equation}
Eliminating $t$ and putting $\varepsilon_2=0$, we see that in this particular case the parametric surface \eqref{eq_12} splits into three surfaces
\begin{eqnarray}
& & \varepsilon_1^2g+k-\varepsilon_1^4c^2-
\frac{1}{4}[h^2+2\varepsilon_1^2(a^2+b^2)h+\varepsilon_1^4(a^2-b^2)^2]=0,\label{eq_13}\\[3mm]
& & 4k-\varepsilon_1^2\{2h(a^2+b^2)-4g+\varepsilon_1^2[4c^2+(a^2-b^2)^2]\}=0,\label{eq_14}
\end{eqnarray}
and
\begin{equation}\label{eq_15}
\begin{array}{l}
4(c^2-a^2b^2)k+[abh+g+\varepsilon_1^2ab(a+b)^2][abh-g+\varepsilon_1^2ab(a-b)^2]-4\varepsilon_1^4c^4\\[3mm]
\quad -\{h^2+2\varepsilon_1^2[h(a^2+b^2)-2g]+\varepsilon_1^4[(a^2-b^2)^2-4a^2b^2]\}c^2=0.
\end{array}
\end{equation}

Consider the second system \eqref{eq_11} in the domain $z\in \overline{{\mathbb C}}\setminus 0$ and introduce $s$ as a root of the equation
\begin{equation*}
s^2-2\{\varepsilon_1^2[\varepsilon_1^2(a^2+b^2)+h+\lambda^2]+2\varepsilon_2^2\}s+4\lambda^2\varepsilon_1^6(a^2+b^2)+8\lambda^2\varepsilon_1^4z^2=0.
\end{equation*}
Then we get the parametric equations for $h,k,g$
\begin{equation}\label{eq_16}
\begin{array}{l}
\ds{g(s)=-\frac{\varepsilon_1^4\lambda^2[(a^2-b^2)^2+4c^2]}{s} -
\frac{\{\varepsilon_1^2[\varepsilon_1^2(a^2+b^2)+h+\lambda^2]+2\varepsilon_2^2\}}{8\varepsilon_1^8\lambda^2}s^2}\\[3mm]
\ds{\qquad+\frac{1}{16}\frac{s^3}{\lambda^2\varepsilon_1^8}+\frac{1}{2}\{\varepsilon_1^2[(a^2-b^2)^2+4c^2]+(a^2+b^2)(h+\lambda^2)\},}\\[5mm]
\ds{k(s)=\frac{\varepsilon_1^8\lambda^4[(a^2-b^2)^2+4c^2]}{s^2}+
 \frac{\{\varepsilon_1^2[\varepsilon_1^2(a^2+b^2)+h+\lambda^2]+2\varepsilon_2^2\}}{2\varepsilon_1^4}s}\\[3mm]
\ds{\qquad-\frac{3}{16}\frac{s^2}{\varepsilon_1^4}-\frac{1}{2}\lambda^2\bigl[2\varepsilon_1^2(a^2+b^2)+h+\frac{\lambda^2}{2}\bigr].}
\end{array}
\end{equation}

If we eliminate $s$ from \eqref{eq_16} and take the value of the gyrostatic momentum $\lambda=0$ keeping arbitrary values of the deformation parameters $\varepsilon_1$ and $\varepsilon_2$, then, similar to the previous case, the parametric surface \eqref{eq_16} splits into three surfaces
\begin{eqnarray}
& & k=0,\label{eq_17}\\[3mm]
& & \{\varepsilon_1^2[\varepsilon_1^2(a^2+b^2)+h]+2\varepsilon_2^2\}^2-4\varepsilon_1^4k=0,\label{eq_18}
\end{eqnarray}
and
\begin{equation*}
\{\varepsilon_1^2[(a^2-b^2)^2+4c^2]+(a^2+b^2)h-2g\}^2-4[(a^2-b^2)^2+4c^2]k=0.
\end{equation*}

In what follows we study the critical subsystems generated by the surfaces $\Pi_{{\cal L}_1}$ and $\Pi_{{\cal L}_2}$ defined by equations \eqref{eq_13} and \eqref{eq_14}, respectively.

\section{New invariant relations}
From now on we consider the integrals $H=H_{\varepsilon_1,0}$, $K=K_{\varepsilon_1,0}$, and $G=G_{\varepsilon_1,0}$ for which the deformation parameter $\varepsilon_2$ vanishes.

Consider the functions
\begin{equation*}
{\cal L}_1(h,k,g)=\varepsilon_1^2g+k-\varepsilon_1^4c^2-
\frac{1}{4}[h^2+2\varepsilon_1^2(a^2+b^2)h+\varepsilon_1^4(a^2-b^2)^2]
\end{equation*}
and
\begin{equation*}
{\cal L}_2(h,k,g)=4k-\varepsilon_1^2\{2h(a^2+b^2)-4g+\varepsilon_1^2[4c^2+(a^2-b^2)^2]\}.
\end{equation*}
These functions specify equations \eqref{eq_13} and \eqref{eq_14} of two-dimensional surfaces $\Pi_{{\cal L}_1}$ and $\Pi_{{\cal L}_2}$ of the type \eqref{eq_5}. The choice of functions ${\cal L}_i(h,k,g)$ is motivated above by singularities of the algebraic curve.

As in \eqref{eq_6}, we determine the corresponding integrals $\Phi_{{\cal L}_1}$ and $\Phi_{{\cal L}_2}$ by the formulas
\begin{equation*}
\Phi_{{\cal L}_1}=\varepsilon_1^2G+K-\varepsilon_1^4({\boldsymbol\alpha}\cdot{\boldsymbol\beta})^2-
\frac{1}{4}[H^2+2\varepsilon_1^2({\boldsymbol\alpha}^2+{\boldsymbol\beta}^2)H+\varepsilon_1^4({\boldsymbol\alpha}^2-{\boldsymbol\beta}^2)^2],
\end{equation*}
and
\begin{equation*}
\Phi_{{\cal L}_2}=4K-\varepsilon_1^2\{2H({\boldsymbol\alpha}^2+{\boldsymbol\beta}^2)-4G+\varepsilon_1^2[4({\boldsymbol\alpha}\cdot{\boldsymbol\beta})^2+
({\boldsymbol\alpha}^2-{\boldsymbol\beta}^2)^2]\}.
\end{equation*}

\begin{propos}
On $\cal P$, the integral $\Phi_{{\cal L}_1}$ can be represented as the product of two polynomials in the phase variables
\begin{equation*}
\Phi_{{\cal L}_1}=F_1\cdot F_2,
\end{equation*}
where
\begin{equation}\label{eq_19}
F_1=M_3+\lambda+\varepsilon_1(\beta_1-\alpha_2),\\[3mm]
\end{equation}
\begin{equation}\label{eq_20}
\begin{array}{l}
\ds{F_2=M_3^3+[\varepsilon_1(\beta_1-\alpha_2)+\lambda]M_3^2+[M_1^2+M_2^2+2\varepsilon_1(\alpha_3M_2-\beta_3M_1)]M_3}\\[3mm]
\ds{\qquad+[(M_2^2-M_1^2)(\alpha_2+\beta_1)+2M_1M_2(\alpha_1-\beta_2)]\varepsilon_1+\lambda(M_1^2+M_2^2).}
\end{array}
\end{equation}
\end{propos}

\begin{remark}
For the Kirchhoff equations on the coalgebra $e(3)^*$ and the Poincare equations on
the coalgebra $so(4)^*$, the additional integral is also represented as a product of two polynomials
\cite{Sok2001,Sok02,Sok03} and \cite{BorMamSok2001}.
\end{remark}

\begin{theorem}
The zero level of each of the functions \eqref{eq_19}, \eqref{eq_20} is an invariant five-dimensional
manifold in $\cal P$.
\end{theorem}
\begin{proof} The derivatives of the functions \eqref{eq_19}, \eqref{eq_20}
in virtue of \eqref{eq_1} have the form
\begin{equation}\label{eq_21}
\dot{F_1} = 2\varepsilon_1(\alpha_1+\beta_2) F_1, \quad \dot{F_2}=-2\varepsilon_1(\alpha_1+\beta_2)F_2.
\end{equation}
It is easy to check that zero is a regular value for both functions. Then \eqref{eq_21} yields that each of the equations $F_k=0$ ($k=1,2$) specifies an invariant five-dimensional manifold in $\cal P$.
\end{proof}

Since invariant five-dimensional submanifolds $\{F_k=0\}$ considered separately does not belong completely to the critical set of the momentum mapping ${\cal F}$, we will consider their intersection given by the system of equations
\begin{equation}\label{eq_22}
F_1=0,\quad F_2=0.
\end{equation}
This system specifies an invariant four-dimensional submanifold ${\cal M}_{{\cal L}_1}$ in
$\cal P$ and it is a critical subsystem of the zero level of the integral $\Phi_{{\cal L}_1}$ since
\begin{equation*}
\Phi_{{\cal L}_1}=0,\quad d\Phi_{{\cal L}_1}=F_1dF_2+F_2dF_1=0.
\end{equation*}
According to the general notation, the manifold (\ref{eq_22}) is denoted by ${\cal M}_{{\cal L}_1}$.
\begin{theorem}
The function
\begin{equation*}
\begin{array}{l}
\ds{F_0=2\varepsilon_1\{-(\beta_2+\alpha_1)M_3^2-(\alpha_1+\beta_2)[\varepsilon_1(\beta_1-\alpha_2)+\lambda]M_3-\alpha_1M_1^2-\beta_2M_2^2}\\[3mm]
\ds{\qquad-(\alpha_2+\beta_1)M_1M_2+[\varepsilon_1(2\alpha_1\beta_3-\beta_1\alpha_3-\alpha_2\alpha_3)+\lambda\alpha_3]M_1-}\\[3mm]
\ds{\qquad-[\varepsilon_1(2\beta_2\alpha_3-\beta_1\beta_3-\alpha_2\beta_3)-\lambda\beta_3]M_2\}}
\end{array}
\end{equation*}
is a first integral of the critical subsystem ${\cal M}_{{\cal L}_1}$.
\end{theorem}
The proof is by straightforward calculation: $\dot{F_0}=\{H,F_0\}=0$.

Note that
\begin{equation*}
\{F_1,F_2\}=F_0.
\end{equation*}
According to Lemma \ref{lem1_1}, this implies that the zero level of the integral $F_0$ is the set of points of co-dimension $1$ of degeneration of the $2$-form induced on ${\cal M}_{{\cal L}_1}$ by the symplectic structure of ${\cal P}$. For this set we have
\begin{equation}\label{eq_23}
F_0=0,\quad F_1=0,\quad F_2=0.
\end{equation}
It easily follows from \eqref{eq_23} that the corresponding values of the first integrals are
\begin{equation}\label{eq_24}
g=0,\quad k=\varepsilon_1^4 c^2+ \frac{1}{4}[h^2+2\varepsilon_1^2(a^2+b^2)h+\varepsilon_1^4(a^2-b^2)^2].
\end{equation}
Obviously, the points (\ref{eq_24}) form the tangency line of the surfaces \eqref{eq_13} and \eqref{eq_15}.

Thus, the following theorem holds.
\begin{theorem}
The phase space of the critical subsystem ${\cal M}_{{\cal L}_1}$ specified by relations \eqref{eq_22} is almost everywhere a four-dimensional submanifold in ${\cal P}$. Moreover, the induced dynamical system is almost everywhere Hamiltonian with two degrees of freedom.
The Hamiltonian function $H$ and the function $F_0$ can be taken as independent integrals for this system.
\end{theorem}

Let us emphasize that by virtue of Lemma \ref{lem1_2}, the set ${\cal M}_{{\cal L}_1}$ given by relations \eqref{eq_22} consists of critical points of the momentum mapping $\cal F$.

Recall that $\varepsilon_2=0$. Then we can verify directly that the system of equations
\begin{equation*}
\Phi_{{\cal L}_2}=0,\quad d\Phi_{{\cal L}_2}=0
\end{equation*}
is equivalent to
\begin{equation}\label{eq_25}
M_1=0,\quad M_2=0.
\end{equation}
So the invariant set ${\cal M}_{{\cal L}_2}$ in $\cal P$ is given by the system (\ref{eq_25}). The fact that this system is a system of invariant relations can be easily checked by a simple calculation:
\begin{equation*}
\begin{array}{l}
\dot{M_1}=\{H,M_1\}=-2\varepsilon_1\beta_2M_1+[2M_3+2\lambda+2\varepsilon_1(\beta_1-\alpha_2)+2\varepsilon_1\alpha_2]M_2,\\[3mm]
\dot{M_2}=\{H,M_2\}=[-2M_3-2\lambda-2\varepsilon_1(\beta_1-\alpha_2)+2\varepsilon_1\beta_1]M_1-2\varepsilon_1\alpha_1M_2,\\[3mm]
\{M_1,M_2\}=M_3.
\end{array}
\end{equation*}
Therefore (\ref{eq_25}) implies
\begin{equation*}
\dot{M_1}=\dot{M_2}=0.
\end{equation*}
Obviously, the functions $M_1,M_2$ are independent on $\cal P$. Thus, using Lemma \ref{lem1_2}, we can conclude that the set ${\cal M}_{{\cal L}_2}$ specified by the system of equations \eqref{eq_25} is a smooth four-dimensional invariant submanifold in $\cal P$ and consists of critical points of the momentum mapping ${\cal F}$. The 2-form induced on ${\cal M}_{{\cal L}_2}$ by the symplectic structure of $\cal P$ degenerates on the set $M_3=0$ of codimension 1.
Thus, the following theorem holds.
\begin{theorem}
The phase space of the critical subsystem ${\cal M}_{{\cal L}_2}$ specified by relation \eqref{eq_25} is a four-dimensional submanifold in ${\cal P}$. The induced dynamical system on it is almost everywhere Hamiltonian with two degrees of freedom.
\end{theorem}

\section{Applications}
Here we show how the systems of invariant relations \eqref{eq_22} and \eqref{eq_25} can be used to determine the type of a critical point $x_0$ of range $2$ in the integrable system with three degrees of freedom in the sense of definition \cite{BolFom}.
Here we follow the scheme suggested in \cite{KhRyabSm2011,KhRCDNew}. In particular, the type of a critical point of an integrable system gives the complete information about the stability of a trajectory passing through this point.

Consider the first integral $\Phi_{\cal L}$ such that it is regular in the neighbourhood of a point $x_0\in{\cal M}_{\cal L}$ except for $x_0$ itself, $d\Phi_{\cal L}(x_0) = 0$. In this case the point $x_0$ appears to be fixed for the Hamiltonian field $\sgrad \Phi_{\cal L}$, and we can find a linearization of this field at the point $x_0$. This linearization is the symplectic operator  $A_{\Phi_{\cal L}}$ in the six-dimensional space that is tangent to the phase space at the point $x_0$. This operator has four zero eigenvalues and the remaining factor of the characteristic polynomial has the form $\mu^2 - C_{\Phi_{\cal L}}$. If $C_{\Phi_{\cal L}} < 0$, we get the point of ``center'' type (the corresponding two-dimensional torus is elliptic and is a stable manifold in the phase space, it is a limit of concentric family of three-dimensional regular tori). If $C_F > 0$, we get the ``saddle'' type point (the corresponding two-dimensional torus is hyperbolic and there are trajectories that are asymptotic to this torus and lie on the three-dimensional separatrix surfaces).

In our problem, the situation is more complicated due to the fact that the phase space is specified in ${\mathbb R}^9$ by three implicit equations \eqref{eq_4} and it is rather difficult to calculate the restrictions of the operators to the tangent surfaces. However, the functions in the left-hand sides of equations \eqref{eq_4} are the Casimir functions for the natural extension to ${\mathbb R}^9$ of the Poisson bracket for the symplectic structure of the space ${\cal P}$. Consequently, when calculating the symplectic operator of the form $A_{\Phi_{\cal L}}$, they will add three zero roots to the characteristic polynomial which has the ninth degree. Thus, we know in advance that under the condition $\sgrad \Phi_{\cal L} = 0$ the required coefficient $C_{\Phi_{\cal L}}$ is a coefficient at $\mu^7$ in the characteristic polynomial $Z_{\Phi_{\cal L}}(\mu)$ of the operator $A_{\Phi_{\cal L}}$ in ${\mathbb R}^9$. Calculating the characteristic polynomial according to the method suggested in \cite{KhRyabSm2011,KhRCDNew} we obtain
\begin{equation*}
Z_{\Phi_{\cal L}} (\mu) = -\mu^7(\mu^2 - C_{\Phi_{\cal L}} ),
\end{equation*}
where
\begin{equation*}
C_{\Phi_{\cal L}}=\frac{1}{2}\trace(A_{\Phi_{\cal L}}^2).
\end{equation*}
Note that the operator $A_{\Phi_{\cal L}}$ is well defined even for for degenerate Poisson bracket, so it is calculated in the space ${\mathbb R}^9$ in the bracket \eqref{eq_3}.

\begin{theorem}
At the points of the critical subsystem ${\cal M}_{{\cal L}_1}$ the coefficient of the characteristic polynomial $C_{\Phi_{{\cal L}_1}}$ is specified by the formula
\begin{equation*}
C_{\Phi_{{\cal L}_1}}=f_0^2,
\end{equation*}
where $f_0$ is the constant of the additional integral $F_0$.
\end{theorem}
Thus, any two-dimensional torus $\{({\boldsymbol M},{\boldsymbol \alpha},{\boldsymbol\beta})\in{\cal M}_{\cal L} :H=h, F_0=f_0\}$ has the hyperbolic type except for the zero value of the additional integral $F_0$ when the corresponding critical points of rank 2 become degenerate.

\begin{theorem}
At the points of the critical subsystem ${\cal M}_{{\cal L}_2}$ the coefficient of the characteristic polynomial $C_{\Phi_{{\cal L}_2}}$ is specified by the following expression:
\begin{equation*}
C_{\Phi_{{\cal L}_2}}=16h\{h[\varepsilon_1^2(a^2+b^2)-\lambda^2]-2\varepsilon_1^2g\}.
\end{equation*}
\end{theorem}
The coefficient $C_{\Phi_{{\cal L}_2}}$ vanishes in the preimage of the tangency line of the surfaces \eqref{eq_14} and \eqref{eq_16} with $\varepsilon_2=0$ in the equation of the surface \eqref{eq_16}.
Denote
\begin{equation*}
  \ds h_*= \frac{2\varepsilon_1^2g}{\varepsilon_1^2(a^2+b^2)-\lambda^2}.
\end{equation*}
Then, the points of the critical subsystem ${\cal M}_{{\cal L}_2}$ have the elliptic type for $h\in(0,h_*)$ and the hyperbolic type  $h\notin[0,h_*]$. For the boundary values $h=0$ and $h=h_*$ the points on the correspondent two-dimensional critical tori are degenerate as the rank 2 critical points of the initial system.

\section{Conclusion}

In the paper, for the problem of the generalized two-field gyrostat motion under gyroscopic forces without linear potential we have found new critical subsystems ${\cal M}_{{\cal L}_1}$ and ${\cal M}_{{\cal L}_2}$ which are almost everywhere smooth four-dimensional manifolds. These subsystems are described in two ways. First, they are defined as the sets of critical points of the integral mapping lying on the zero level of some naturally arising general first integral. Second, the phase spaces of the critical subsystems are described by the pair of invariant relations in \eqref{eq_22} and \eqref{eq_25}. For each critical subsystem, the obtained integral and the system of invariant relations provide a way
to explicitly calculate the type of the corresponding critical points with respect to the initial integrable system with three degrees of freedom.

If the deformation parameters $\varepsilon_1$ and $\varepsilon_2$ are different from zero, but the gyrostatic momentum $\lambda$ vanishes, the bifurcation surfaces specified by \eqref{eq_17} and \eqref{eq_18} also give rise to almost everywhere invariant four-dimensional submanifolds. In \cite{Rya13} the explicit equations of these four-dimensional submanifolds are suggested. The problem of determining the invariant relations corresponding to the parametric surface \eqref{eq_16} for which the deformation parameters $\varepsilon_1$, $\varepsilon_2$ and the parameter of the gyrostatic momentum $\lambda$ are different from zero still remains unsolved, so the types of the critical points in this case are not completely established.

The author thanks Professor M.\,P.\,Kharlamov for valuable discussions and constant help in preparing this article.

\end{document}